\documentclass[11pt,letterpaper]{article}
\usepackage[T1]{fontenc}
\usepackage[utf8]{inputenc}
\usepackage{lmodern}
\usepackage[margin=1.1in]{geometry}
\usepackage{amsmath,amssymb,amsthm,mathtools}
\usepackage{microtype}
\usepackage{enumitem}
\usepackage{url}
\numberwithin{equation}{section}
\newtheorem{theorem}{Theorem}[section]
\newtheorem{lemma}[theorem]{Lemma}
\newtheorem{proposition}[theorem]{Proposition}
\newtheorem{corollary}[theorem]{Corollary}
\theoremstyle{definition}
\newtheorem{definition}[theorem]{Definition}
\newtheorem{example}[theorem]{Example}
\newtheorem{remark}[theorem]{Remark}
\newcommand{\F}{\mathbb F}
\newcommand{\Q}{\mathbb Q}
\newcommand{\R}{\mathbb R}
\newcommand{\C}{\mathbb C}
\newcommand{\Z}{\mathbb Z}
\newcommand{\OO}{\mathcal O}
\newcommand{\cD}{\mathcal D}
\newcommand{\cP}{\mathcal P}
\newcommand{\Fpbar}{\overline{\F}_p}
\DeclareMathOperator{\Gal}{Gal}

\title{Nonlinear fixed-distance codes via\\real lifts of moment configurations}
\author{Sergey Yekhanin\\[2pt]\small Microsoft Research\\\small\texttt{yekhanin@microsoft.com}}

\begin{document}
\maketitle

\begin{abstract}
Let $A_q(n,d)$ denote the largest size of a code of length $n$ and minimum Hamming distance $d$ over an alphabet of $q$ symbols. For fixed $q\ge2$ and $d\ge4$ we construct codes with redundancy
$n-\log_q A_q(n,d)\le(d-3)\log_q n+O(\sqrt{\log n})$ for all large $n$.
No algebraic structure is imposed on the alphabet, and the codes are not linear. For $d=4$ and $d=5$ the coefficient of $\log_q n$ matches the Hamming bound, so the redundancy is determined up to lower-order terms. At distance four the codes are provably better than every linear code: by the polynomial-method bound for cap sets, $\F_q$-linear codes with $q>2$ require a strictly larger coefficient in front of $\log_q n$.

The construction imposes moment checksums over an auxiliary finite field, as BCH codes do, but one fewer than the distance would normally require. The affine-line theorem from prior work describes the dependencies left by these checks. We exclude them by placing the code positions inside a specially structured subset of the field, built from Behrend's sphere construction with the help of algebraic number theory.
\end{abstract}

\noindent\textbf{Keywords.} Nonlinear codes, fixed minimum distance, redundancy, cap sets, Behrend construction, number fields, Chebotarev density theorem.

\section{Introduction}\label{sec:introduction}
Let $A_q(n,d)$ denote the largest size of a code of length $n$ and minimum Hamming distance at least $d$ over an alphabet of $q$ symbols. Throughout, $q\ge2$ and $d$ are fixed integers and $n\to\infty$. The alphabet is $\Sigma_q=\{0,1,\ldots,q-1\}$; $q$ is not assumed to be a prime power, and codes are not assumed to be linear. We measure a code by its redundancy and write
\[
 \rho_q(n,d)=n-\log_q A_q(n,d),\qquad
 c(q,d)=\liminf_{n\to\infty}\frac{\rho_q(n,d)}{\log_q n}.
\]
Unsubscripted logarithms are natural.

The fixed-distance problem asks how many check symbols are needed to protect a long word against a fixed number of errors. With $t=\lfloor(d-1)/2\rfloor$, the Hamming bound
\begin{equation}\label{eq:hamming}
 A_q(n,d)\sum_{j=0}^{t}\binom{n}{j}(q-1)^j\le q^n
\end{equation}
gives $\rho_q(n,d)\ge t\log_q n-O_{q,d}(1)$. Over the binary alphabet, extended BCH codes attain the coefficient $t$ for every fixed $d$. Over larger alphabets the picture is less complete \cite{YD}, and that is the setting of this paper. Our main result imposes no algebraic structure on the alphabet.

\begin{theorem}\label{thm:main}
For every fixed pair of integers $q\ge2$ and $d\ge4$ there are constants $C_{q,d}>0$ and $n_{q,d}$ such that for all $n\ge n_{q,d}$,
\begin{equation}\label{eq:main}
 A_q(n,d)\ge\frac{q^n}{n^{d-3}\exp\bigl(C_{q,d}\sqrt{\log n}\bigr)}.
\end{equation}
In particular, $c(q,d)\le\limsup_{n\to\infty}\rho_q(n,d)/\log_q n\le d-3$.
\end{theorem}

\begin{corollary}\label{cor:fourfive}
For every fixed integer $q\ge2$,
\[
 \rho_q(n,4)=\log_q n+O_q(\sqrt{\log n}),\qquad
 \rho_q(n,5)=2\log_q n+O_q(\sqrt{\log n}),
\]
so the normalized redundancies at distances four and five converge to one and two.
\end{corollary}

Some words on what these statements do and do not say. The optimality in Corollary~\ref{cor:fourfive} is optimality of the coefficient of $\log_q n$. The construction loses a factor $\exp(O_{q,d}(\sqrt{\log n}))=n^{o(1)}$ in the number of codewords, so it does not give redundancy $t\log_q n+O(1)$, positive packing density, or asymptotically perfect codes. At distance six, Theorem~\ref{thm:main} gives coefficient three where the Hamming bound has two, and as $d$ grows the gap between $d-3$ and the Hamming coefficient $\lfloor(d-1)/2\rfloor$ grows linearly.

\subsection{Why nonlinearity matters at distance four}
For a prime power $q$, let $\rho_q^{\rm lin}(n,d)$ be the smallest redundancy $n-k$ of an $\F_q$-linear code of length $n$ and minimum distance at least $d$, and let
\begin{equation}\label{eq:theta}
 \theta_q=\min_{0<x<1}\frac{1+x+\cdots+x^{q-1}}{x^{(q-1)/3}}.
\end{equation}
The Ellenberg--Gijswijt theorem \cite{EG} has the following consequence, proved in Section~\ref{sec:four}.

\begin{proposition}\label{prop:linear}
For every prime power $q>2$, one has $\theta_q<q$ and
\begin{equation}\label{eq:linear}
 \rho_q^{\rm lin}(n,4)\ge\frac{\log q}{\log\theta_q}\log_q n-O_q(1).
\end{equation}
For $q=3$ the coefficient is approximately $1.084026$.
\end{proposition}

Together with Corollary~\ref{cor:fourfive}, this shows that for every prime power $q>2$ nonlinear codes achieve a leading redundancy coefficient that no linear family can. The ternary case is the most transparent: the columns of a parity-check matrix of a distance-four code form a subset of $\F_3^{n-k}$ with no three-term arithmetic progression, and such sets have exponentially small density, whereas Behrend-type sets of integers lose only a subpolynomial factor.

\subsection{Prior work and what is new}
Extended BCH codes provide the classical benchmark. For a prime power $q$, their redundancy satisfies
\begin{equation}\label{eq:bch}
 \rho_q^{\rm lin}(n,d)\le
 \left\lceil\frac{(d-2)(q-1)}q\right\rceil\log_q n+O_{q,d}(1);
\end{equation}
see \cite[eq.~(3)]{YD}. For binary codes this meets the leading term of the Hamming bound. For alphabets large relative to the distance, namely $q\ge d-1$, the BCH coefficient is $d-2$, and Theorem~\ref{thm:main} lowers it to $d-3$. For smaller alphabets, BCH codes can match or improve on our bound. Combining the two constructions gives
\[
 c(q,d)\le\min\left\{d-3,\left\lceil\frac{(d-2)(q-1)}q\right\rceil\right\}
\]
for prime powers $q$.

Yekhanin and Dumer \cite{YD} developed the systematic study of the redundancy coefficient $c(q,d)$ and proved
\begin{equation}\label{eq:yd}
 c(q,d)\le(d-3)+\frac1{d-2}
\end{equation}
for all $q\ge2$ and $d\ge3$. For $d\ge4$, our result removes the additional $1/(d-2)$. At distances four and five this reaches the Hamming coefficient over every alphabet.

Some of these coefficients were already attainable in special cases. Dumer's earlier constructions \cite{Dumer} give coefficient $7/3$ at distance five and $3$ at distance six, so the distance-six coefficient in Theorem~\ref{thm:main} is not new. At distance five, ternary and quaternary linear codes already attain coefficient two \cite[Section~II]{YD}. Liu and Xing \cite{LX} construct nonlinear codes with redundancy $(d-2)\log_q n+O_{q,d}(1)$ and obtain finite-length improvements over linear codes.

Our contribution is a construction giving coefficient $d-3$ for every alphabet and every fixed distance $d\ge4$. It builds on the algebraic structure identified by Yekhanin and Dumer, but uses a geometric restriction in place of an additional check. The resulting codes are nonlinear; at distance four, Proposition~\ref{prop:linear} shows that this is essential to the improvement over linear codes. The next subsection explains the construction and the role of the geometry.

\subsection{Overview of the method}
The basic idea is to impose fewer checks, but to be more selective about the field elements used to label the code positions. The checks leave a restricted family of troublesome configurations. We then choose a large set of labels containing none of them.

\smallskip\noindent\textbf{One fewer moment check.}
Assign a distinct element $x_i$ of an auxiliary field $\F_{p^m}$ to each code position; these elements are the \emph{locators}. Read each symbol $u_i\in\{0,\ldots,q-1\}$ as an integer. In addition to a symbol-sum check, record the $d-3$ moments
\[
 \sum_i u_i x_i,\quad \sum_i u_i x_i^2,\quad\ldots,\quad
 \sum_i u_i x_i^{d-3}.
\]
Words with the same check values form a code. The auxiliary field need not have any relation to the alphabet, so this construction makes sense even when $q$ is not a prime power.

For sufficiently large characteristic, the usual Vandermonde argument shows that two distinct words with the same check values cannot differ in $d-2$ or fewer positions. To obtain distance $d$, we must also rule out differences in exactly $d-1$ positions. One more moment check would do this, but would cost another leading term of $\log_q n$ in the redundancy. Instead, we restrict the locators to a set $S\subseteq\F_{p^m}$ on which such a difference cannot occur.

\smallskip\noindent\textbf{What can still go wrong.}
A difference in $d-1$ positions that escapes the checks gives a short relation among the corresponding locators; we call it a \emph{dependency}. Yekhanin and Dumer \cite{YD} show that, for suitable choices of $p$ and $m$, all locators in such a dependency lie on one affine line over $\F_p$. Here we view $\F_{p^m}$ as an $m$-dimensional vector space over $\F_p$.

There is also strong control over the relative positions of these locators. After translating and scaling so that two of them become $0$ and $1$, the equations for a dependency leave only finitely many possible configurations of distinct locators. These normalized configurations are the \emph{shapes} studied later. For our purposes, it is enough to control just one ratio,
\[
 \lambda=\frac{x_3-x_1}{x_2-x_1}.
\]
The equations for the shapes restrict this ratio to roots modulo $p$ of finitely many fixed polynomials, depending only on $q$ and $d$. Thus the task is not to exclude arbitrary finite-field collinearity, but only triples with these particular ratios.

\smallskip\noindent\textbf{Why a sphere helps.}
Over the real numbers, there is a simple way to exclude all such triples at once: put the points on a sphere. A line meets a sphere in at most two points, so three distinct points of the sphere cannot satisfy $w=(1-\lambda)u+\lambda v$ for any real $\lambda$. Behrend's construction \cite{Behrend,EHPS} uses this observation to build large sets of integers containing no three-term arithmetic progression. We adapt the same idea to the ratios left by our checks.

The difficulty is that our locators live in a finite field, not in real space. A relation modulo $p$ need not have held before reduction, and a root of a polynomial modulo $p$ need not correspond to a real root of that polynomial. The number theory in the paper addresses precisely these two obstacles.

\smallskip\noindent\textbf{Bringing the real geometry into the finite field.}
First, we choose the auxiliary primes using Chebotarev's theorem so that every ratio arising from a dependency is the reduction of one of finitely many fixed real algebraic numbers. We call such a number a \emph{real lift} of the ratio. Only the ratio of one triple needs a real lift; we do not need to lift the entire configuration.

Next, we work in a fixed number field contained in $\R$: a finite extension of $\Q$ containing these coefficients. We build the sphere using algebraic integers as coordinates and keep them small enough that the relevant congruences must already be equalities before reduction. The proof is an algebraic version of the elementary observation that an integer divisible by $p$ and of absolute value less than $p$ must be zero. Consequently, a prohibited triple in the reduced set would lift to three distinct collinear points on the real sphere, which is impossible.

\smallskip\noindent\textbf{The cost of restricting the locators.}
As in Behrend's construction, the sphere can be chosen to contain many points. With suitable parameters, we obtain $n$ locators in a field of size $n\exp(O_{q,d}(\sqrt{\log n}))$. The $d-3$ moment checks therefore contribute the leading redundancy $(d-3)\log_q n$; the symbol-sum check and the restriction of the locator set together cost only $O_{q,d}(\sqrt{\log n})$. In this way, the geometric restriction replaces an entire moment check at only a lower-order cost.

At distance four, the argument can be carried out entirely over the integers. At distance five, the relevant ratios satisfy quadratic equations, and an elementary choice of primes suffices. A shorter proof of the distance-five bound would follow from a Sidon-type estimate recorded by Cilleruelo and Timmons \cite[p.~2]{CT}, but the required form of that estimate is attributed there to unpublished work and personal communication. We therefore give a direct proof that does not rely on it; see Remark~\ref{rem:sidon} for details. The number-field construction allows the same geometric idea to work for general distance.

\subsection{Organization}
Section~\ref{sec:four} establishes Theorem~\ref{thm:main} for $d=4$ and proves Proposition~\ref{prop:linear}. Section~\ref{sec:moments} sets up the moment checksums, records the finiteness of normalized configurations, and applies the Yekhanin--Dumer affine-line theorem. Section~\ref{sec:sphere} develops the sphere construction over a number field and reduces the main theorem to a statement about affine ratios that we call a localization datum. Section~\ref{sec:five} produces such a datum at distance five by elementary means. Section~\ref{sec:general} does so in general, using one-variable polynomial relations and Chebotarev's theorem. Section~\ref{sec:limitations} discusses limitations and questions.

\subsection{Acknowledgement and disclosure of AI use}
The author shared with GPT-6 (OpenAI) his research notes outlining possible approaches to the fixed-distance coding problem. Starting from those notes, GPT-6 developed the constructions and the proofs presented in the paper, and also assisted with the literature review, with checking the mathematics, and with preparing the manuscript. This was a substantive research contribution far beyond language editing. GPT's original approach has been simplified by the author. Claude Fable 5.1 (Anthropic) was used to revise the exposition and presentation. Sergey Yekhanin is the sole named author and takes responsibility for the final version.

\section{Distance four}\label{sec:four}
At distance four the construction is elementary, and it already shows the two ingredients of the general argument: a checksum that can be lifted from a congruence to an exact integer equality, and a locator set built from Behrend's sphere. Put $K=q-1$ and read the symbols of $\Sigma_q$ as integers. When two words differ in at most three positions, each nonzero coordinate difference is then an integer of absolute value at most $K$, and the following notion is designed for exactly those differences.

\begin{definition}\label{def:weighted}
A set $A\subseteq\Z$ is \emph{$K$-weighted-average-free} if, for all positive integers $a,b$ with $a+b\le K$, the only solutions of
\begin{equation}\label{eq:weighted}
 ax+by=(a+b)z,\qquad x,y,z\in A,
\end{equation}
are those with $x=y=z$. For $K=1$ the condition is vacuous.
\end{definition}

\subsection{The construction}
\begin{lemma}\label{lem:fourcode}
Let $A=\{a_1,\ldots,a_n\}\subseteq\{0,\ldots,N-1\}$ be a $K$-weighted-average-free set of $n$ distinct integers, and put $L=3K+1$ and $M=3KN+1$. Every fiber of the map
\begin{equation}\label{eq:fourchecks}
 \Phi(u)=\left(\sum_{i=1}^n u_i\bmod L,\ \sum_{i=1}^n a_i u_i\bmod M\right),
 \qquad u\in\Sigma_q^n,
\end{equation}
is a code of minimum distance at least four. Consequently
$A_q(n,4)\ge q^n/\bigl((3K+1)(3KN+1)\bigr)$.
\end{lemma}
\begin{proof}
Suppose two distinct words in one fiber differ on a support of size $w\le3$, with coordinate differences $\delta_1,\ldots,\delta_w$ at locators $x_1,\ldots,x_w\in A$. Then $\sum_i\delta_i\equiv0\pmod L$ and $\sum_i\delta_i x_i\equiv0\pmod M$. Since $|\sum_i\delta_i|\le3K<L$ and $|\sum_i\delta_i x_i|\le3K(N-1)<M$, both congruences are equalities in $\Z$. For $w=1$ the first is impossible. For $w=2$ the two together give $\delta_1(x_1-x_2)=0$, so $x_1=x_2$. For $w=3$, after reordering and possibly changing all signs, the coefficients are $(a,b,-a-b)$ with $a,b>0$ and $a+b\le K$, and the second equality is a relation \eqref{eq:weighted} on three distinct elements of $A$, which does not exist. The number of fibers is at most $LM$.
\end{proof}

The modulus $L$ is a constant because it only has to dominate the symbol sum of a short difference of two words, not of a whole word; this is why the first checksum costs only $O_q(1)$ redundancy. The code has redundancy $\log_q N+O_q(1)$, and everything depends on how small $N$ can be made for a given $n$.

\begin{lemma}\label{lem:behrend}
For integers $K\ge1$, $h\ge2$ and $m\ge1$ there is a $K$-weighted-average-free set $A\subseteq\{0,\ldots,(Kh)^m-1\}$ with
\[
 |A|\ge\frac{h^m}{m(h-1)^2+1}\ge\frac{h^{m-2}}m.
\]
\end{lemma}
\begin{proof}
The construction is Behrend's: points on a sphere, written out in a large base. We go through it in four steps.

\smallskip\noindent\textbf{A large sphere.}
For a point $v=(v_0,\ldots,v_{m-1})$ of the grid $\{0,\ldots,h-1\}^m$, the squared Euclidean norm $\|v\|^2=\sum_j v_j^2$ is an integer between $0$ and $m(h-1)^2$, so it takes at most $m(h-1)^2+1$ values. By the pigeonhole principle some value $\rho^2$ is taken by at least $h^m/(m(h-1)^2+1)$ grid points. Let $T$ be the set of grid points with $\|v\|^2=\rho^2$.

\smallskip\noindent\textbf{Encoding.}
Put $B=Kh$ and $f(v)=\sum_{j=0}^{m-1}v_jB^j$. Since every coordinate satisfies $0\le v_j\le h-1<B$, this is the base-$B$ expansion of $f(v)$; hence $f$ is injective and $0\le f(v)\le B^m-1=(Kh)^m-1$. Let $A=f(T)$, so that $|A|=|T|$.

\smallskip\noindent\textbf{No carries.}
Suppose $af(u)+bf(v)=(a+b)f(w)$ with $u,v,w\in T$ and positive integers $a,b$ with $a+b\le K$. Both sides are sums of powers of $B$,
\[
 af(u)+bf(v)=\sum_{j=0}^{m-1}(au_j+bv_j)B^j,\qquad
 (a+b)f(w)=\sum_{j=0}^{m-1}(a+b)w_jB^j,
\]
whose coefficients $au_j+bv_j$ and $(a+b)w_j$ are nonnegative integers at most $(a+b)(h-1)\le K(h-1)<B$. Base-$B$ expansions with digits in $\{0,\ldots,B-1\}$ are unique, so the coefficients agree for every $j$: $au+bv=(a+b)w$ as vectors. Dividing by $a+b$ and setting $\lambda=a/(a+b)\in(0,1)$,
\[
 w=\lambda u+(1-\lambda)v.
\]

\smallskip\noindent\textbf{Three points of a sphere are not collinear.}
Expanding the square,
\[
 \|\lambda u+(1-\lambda)v\|^2
 =\lambda^2\|u\|^2+2\lambda(1-\lambda)\langle u,v\rangle+(1-\lambda)^2\|v\|^2,
\]
and substituting $2\langle u,v\rangle=\|u\|^2+\|v\|^2-\|u-v\|^2$ gives
\begin{equation}\label{eq:sphereidentity}
 \|\lambda u+(1-\lambda)v\|^2
 =\lambda\|u\|^2+(1-\lambda)\|v\|^2-\lambda(1-\lambda)\|u-v\|^2.
\end{equation}
Apply this with $w=\lambda u+(1-\lambda)v$ and $\|u\|^2=\|v\|^2=\|w\|^2=\rho^2$: the left side is $\rho^2$ and the first two terms on the right add up to $\rho^2$, so $\lambda(1-\lambda)\|u-v\|^2=0$. As $\lambda(1-\lambda)>0$, this forces $u=v$, and then $w=u$ as well. Thus the only solutions of \eqref{eq:weighted} in $A$ have $x=y=z$, and $A$ is $K$-weighted-average-free.

Finally, $m(h-1)^2+1\le mh^2$ because $(h-1)^2+1\le h^2$ for $h\ge1$, which gives the second inequality in the statement.
\end{proof}

\begin{lemma}\label{lem:behrendsize}
For fixed $K\ge1$ and every $n\ge2$ there is a $K$-weighted-average-free set of $n$ integers in $\{0,\ldots,N-1\}$ with $\log N\le\log n+O_K(\sqrt{\log n})$.
\end{lemma}
\begin{proof}
Let $t=\lceil\sqrt{\log_2 n}\rceil\ge1$, $m=t+2$, $h=2^m$ and $N=(Kh)^m$. Lemma~\ref{lem:behrend} gives at least $2^{m(m-2)}/m=2^{t^2+2t}/(t+2)\ge2^{t^2}\ge n$ locators, because $4^t\ge t+2$; keep any $n$ of them. Finally
\[
 \log_2 N=m^2+m\log_2 K=\log_2 n+O_K(\sqrt{\log n}).\qedhere
\]
\end{proof}

\begin{proof}[Proof of Corollary~\ref{cor:fourfive} at distance four]
Lemmas~\ref{lem:fourcode} and~\ref{lem:behrendsize} give codes with redundancy at most
$\log_q N+O_q(1)=\log_q n+O_q(\sqrt{\log n})$, and \eqref{eq:hamming} with $t=1$ gives $\rho_q(n,4)\ge\log_q n-O_q(1)$.
\end{proof}

Two remarks. For $q=2$ the weighted-average-free condition is empty and consecutive integer locators already give redundancy $\log_2 n+O(1)$, so the sphere is needed only for larger alphabets. And even when $q$ is a prime power, the map \eqref{eq:fourchecks} is not an $\F_q$-linear syndrome: it adds integer representatives and reduces modulo two moduli that have nothing to do with $q$. The next subsection shows that this nonlinearity is unavoidable.

\subsection{A linear obstruction}
Let $q$ be a prime power and let $H$ be a full-rank $R\times n$ parity-check matrix of an $\F_q$-linear code of minimum distance at least four. Its columns $h_1,\ldots,h_n$ are distinct and nonzero, and any three of them are linearly independent. Over $\F_3$ this says in particular that no three distinct columns satisfy $h_i+h_j+h_k=0$, which is the equation of a three-term arithmetic progression; since two equal columns in such a relation force the third to equal them, the set of columns is a cap set in $\F_3^R$. Meshulam's bound $O(3^R/R)$ \cite{Meshulam} and the Bateman--Katz bound $O(3^R/R^{1+\varepsilon})$ \cite{BK} already show that $\rho_3^{\rm lin}(n,4)-\log_3 n$ is unbounded, but they are compatible with leading coefficient one. The polynomial-method breakthrough of Croot, Lev, and Pach~\cite{CLP} established an exponential bound for progression-free subsets of $(\mathbb Z/4\mathbb Z)^R$. The subsequent result of Ellenberg and Gijswijt~\cite{EG} rules out leading coefficient one for linear codes over every prime-power alphabet $q>2$.

\begin{proof}[Proof of Proposition~\ref{prop:linear}]
Ellenberg and Gijswijt \cite[Theorem~4]{EG} prove the following. If $\alpha,\beta,\gamma\in\F_q$ satisfy $\alpha+\beta+\gamma=0$ and $\gamma\ne0$, and $S\subseteq\F_q^R$ contains no solution of $\alpha x+\beta y+\gamma z=0$ other than $x=y=z$, then
\begin{equation}\label{eq:eg}
 |S|\le3\,\#\left\{(e_1,\ldots,e_R)\in\{0,\ldots,q-1\}^R:
 \sum_i e_i\le(q-1)R/3\right\}.
\end{equation}
For $0<x<1$, each tuple counted on the right satisfies $x^{\sum_i e_i}\ge x^{(q-1)R/3}$, so the count is at most
\[
 x^{-(q-1)R/3}\sum_e x^{\sum_i e_i}
 =x^{-(q-1)R/3}(1+x+\cdots+x^{q-1})^R.
\]
Minimizing over $x$ gives $|S|\le3\theta_q^R$.

Let $S$ be the set of columns of $H$. As $q>2$ we can choose $\alpha,\beta,\gamma$ all nonzero with zero sum, for instance $(1,a,-1-a)$ with $a\notin\{0,-1\}$. A solution on three distinct columns would contradict the minimum distance, and a solution with exactly two equal columns forces the third to equal them, because the remaining coefficient is nonzero. Hence $n\le3\theta_q^R$, which is \eqref{eq:linear}.

It remains to see that $\theta_q<q$. The function minimized in \eqref{eq:theta} equals $q$ at $x=1$, and the derivative of its logarithm there is $(q-1)/2-(q-1)/3=(q-1)/6>0$, so the function is smaller than $q$ just to the left of $1$. For $q=3$ the minimum is attained at $x=(\sqrt{33}-1)/8$, giving $\theta_3=2.755104613\ldots$ and $\log3/\log\theta_3=1.084026266\ldots$.
\end{proof}

Combined with the construction, Proposition~\ref{prop:linear} gives
\[
 \rho_q^{\rm lin}(n,4)-\rho_q(n,4)
 \ge\left(\frac{\log q}{\log\theta_q}-1\right)\log_q n-O_q(\sqrt{\log n})
\]
for every prime power $q>2$.

\section{Moment checksums and affine-line structure}\label{sec:moments}
For distances five and beyond we use an auxiliary finite field. Fix $q$ and $d$ and put
\[
 K=q-1,\qquad s=d-1,\qquad r=d-3.
\]
Let $p>sK$ be a prime, let $m\ge1$, and let $F=\F_{p^m}$.

\begin{definition}\label{def:dependency}
A \emph{dependency} on distinct elements $x_1,\ldots,x_w\in F$, where $1\le w\le s$, is a relation
\begin{equation}\label{eq:dependency}
 \sum_{i=1}^w\delta_i=0\quad\text{in }\Z,\qquad
 \sum_{i=1}^w\delta_i x_i^j=0\quad\text{in }F\quad(1\le j\le r),
\end{equation}
with integers $\delta_i$ satisfying $0<|\delta_i|\le K$, reduced modulo $p$ in the field equations. A set $S\subseteq F$ is \emph{admissible} if there is no dependency on elements of $S$.
\end{definition}

\begin{lemma}\label{lem:momentcode}
If $S\subseteq F$ is admissible and $n\le|S|$, there is a $q$-ary code of length $n$, minimum distance at least $d$ and cardinality at least $q^n/p^{rm+1}$.
\end{lemma}
\begin{proof}
Choose distinct $x_1,\ldots,x_n\in S$ and partition $\Sigma_q^n$ according to the value of
\begin{equation}\label{eq:momentchecks}
 \Phi(u)=\left(\sum_i u_i\bmod p,\ \sum_i u_i x_i,\ldots,\sum_i u_i x_i^r\right)
 \in\F_p\times F^r.
\end{equation}
If two words in the same fiber differ on a support of size $w\le s$, their coordinate differences $\delta_i$ are nonzero integers of absolute value at most $K$ with $\sum_i\delta_i\equiv0\pmod p$. As $|\sum_i\delta_i|\le sK<p$, the sum is zero in $\Z$, and the moment checks give the field equations in \eqref{eq:dependency}. This is a dependency on elements of $S$, which does not exist. So every fiber has minimum distance at least $d$, and there are at most $p^{rm+1}$ fibers.
\end{proof}

The symbol sum could equally be taken modulo the constant $sK+1$; using $p$ keeps the notation uniform, and the resulting $\log p$ disappears into the lower-order term.

\begin{lemma}\label{lem:vandermonde}
Every dependency involves at least $r+2=d-1$ locators. Moreover, if $x_1,\ldots,x_w$ carry a dependency with coefficients $\delta_1,\ldots,\delta_w$, then so do $a+bx_1,\ldots,a+bx_w$ for any $a,b\in F$ with $b\ne0$.
\end{lemma}
\begin{proof}
The zeroth relation and the $r$ moment relations say together that $\sum_i\delta_i x_i^j=0$ in $F$ for $0\le j\le r$. If $w\le r+1$, the equations with $0\le j\le w-1$ form a linear system in the unknowns $\delta_1,\ldots,\delta_w\in\F_p$ whose matrix is the Vandermonde matrix of the distinct $x_i$. It is invertible, so all $\delta_i$ vanish modulo $p$, which is impossible since $0<|\delta_i|\le K<p$. For the second assertion expand $(a+bx_i)^j$ by the binomial theorem: $\sum_i\delta_i(a+bx_i)^j$ is a linear combination of the sums $\sum_i\delta_i x_i^k$, $0\le k\le j$, all of which vanish.
\end{proof}

So the only dependencies that can occur are those on exactly $r+2$ locators, and they are invariant under affine changes of variable. Given such a dependency, apply $x\mapsto(x-x_1)/(x_2-x_1)$ to send $x_1$ to $0$ and $x_2$ to $1$; the remaining locators become $y_i=(x_{i+2}-x_1)/(x_2-x_1)$ for $1\le i\le r$. Let $\cD$ be the finite set of coefficient tuples $\delta=(\delta_1,\ldots,\delta_{r+2})$ of nonzero integers with $|\delta_i|\le K$ and $\sum_i\delta_i=0$. The normalized dependency says that $(y_1,\ldots,y_r)$ solves the system
\begin{equation}\label{eq:shapes}
 f_j(y_1,\ldots,y_r):=\delta_2+\sum_{i=1}^r\delta_{i+2}y_i^j=0,
 \qquad 1\le j\le r,
\end{equation}
and that $0,1,y_1,\ldots,y_r$ are pairwise distinct. We call such a solution, over any field, a \emph{shape} for $\delta$.

\begin{lemma}\label{lem:finiteshapes}
Let $\delta\in\cD$, and let $k$ be an algebraically closed field of characteristic zero or of characteristic $p>\max(r,K)$. Then every shape for $\delta$ over $k$ is an isolated nonsingular solution of \eqref{eq:shapes}. There are at most $r!$ such shapes. In characteristic zero, all their coordinates are algebraic numbers.
\end{lemma}
\begin{proof}
This is the normalized-moment calculation of \cite[Lemmas~8--9]{YD}, which also works in characteristic zero. The Jacobian matrix has entries $\partial f_j/\partial y_i=j\delta_{i+2}y_i^{j-1}$. Factoring $j$ out of row $j$ and $\delta_{i+2}$ out of column $i$ leaves a Vandermonde matrix, so
\begin{equation}\label{eq:jacobian}
 \det\left(\frac{\partial f_j}{\partial y_i}\right)_{j,i=1}^r
 =r!\left(\prod_{i=1}^r\delta_{i+2}\right)
 \prod_{1\le i<k\le r}(y_k-y_i).
\end{equation}
This is nonzero at every shape. Hence each shape is isolated and nonsingular. The isolated-point form of B\'ezout's bound gives at most $\prod_{j=1}^r\deg f_j=r!$ such solutions.

For algebraicity, put
\begin{equation}\label{eq:Delta}
 \Delta(y)=\prod_{i=1}^r y_i(y_i-1)
 \prod_{1\le i<j\le r}(y_i-y_j).
\end{equation}
Adjoining a variable $z$ and the equation $z\Delta(y)=1$ gives a polynomial system over $\Q$ whose complex solutions are precisely the shapes, with $z=1/\Delta(y)$. They form a finite set. By the standard finiteness theorem for polynomial ideals \cite[Chapter~5, \S3, Theorem~6]{CLO},
\[
 \C[y_1,\ldots,y_r,z]/I\mathrm{\ \ \ with \ \ \ } I=(f_1,\ldots,f_r,z\Delta-1)
\]
is finite-dimensional over $\C$. Hence, for each $i$, sufficiently many powers $1,y_i,y_i^2,\ldots$ are linearly dependent modulo $I$ over $\mathbb{C}$. Since $I$ is generated by polynomials with rational coefficients, such a dependence can be chosen with coefficients in $\mathbb{Q}.$ Thus $y_i$ satisfies a nonzero polynomial in $\mathbb{Q}[T]$, and so every coordinate of every characteristic-zero shape is algebraic over $\mathbb{Q}$.
\end{proof}

The following is the affine-line theorem of Yekhanin and Dumer, in the notation needed here.

\begin{lemma}[Yekhanin--Dumer \cite{YD}]\label{lem:affineline}
Suppose $p>(r+2)K$ and $m>r!$ is prime. Every shape over $\F_{p^m}$ has all its coordinates in $\F_p$. Consequently, all locators of every dependency lie on the affine line
\[
 x_1+(x_2-x_1)\F_p.
\]
\end{lemma}
\begin{proof}
Apply \cite[Lemma~10]{YD} with base field $\F_p$ and $t=r$ to \eqref{eq:shapes}. Its coefficients are nonzero modulo $p$, and $p>r$.
The short argument is as follows. Coordinatewise $p$th powers permute the shapes. A shape in $\F_{p^m}^r$ not fixed by this map has an orbit of length $m$, since $m$ is prime. This contradicts the bound of $r!$ in Lemma~\ref{lem:finiteshapes}. Hence every coordinate is fixed by $a\mapsto a^p$ and lies in $\F_p$. Undoing the normalization gives the stated line; this is \cite[Theorem~4]{YD}.
\end{proof}

The affine-line theorem alone is not enough for our construction: we cannot lift an arbitrary coefficient in $\F_p$ to a bounded real coefficient. What remains is to show that the particular ratio $y_1=(x_3-x_1)/(x_2-x_1)$ is the reduction of a member of a fixed finite set of real algebraic numbers. Section~\ref{sec:sphere} constructs large sets avoiding these ratios. Section~\ref{sec:five} supplies the ratios directly from quadratic equations when $d=5$, and Section~\ref{sec:general} supplies them in general using polynomial identities and Chebotarev's theorem.

\section{A sphere over a number field}\label{sec:sphere}
\subsection{Small algebraic integers}\label{subsec:integers}
A \emph{number field} $E$ is a finite extension of $\Q$. Its degree $D=[E:\Q]$ is the dimension of $E$ as a vector space over $\Q$. Such a field has exactly $D$ embeddings
\[
 \sigma:E\hookrightarrow\C,
\]
where an embedding means an injective field homomorphism that fixes $\Q$. In this section we use only number fields that admit a real embedding. We fix one such embedding and identify $E$ with its image in $\R$; thus we write $E\subseteq\R$. The other embeddings $E\hookrightarrow\C$ need not have real image, so $E$ need not be totally real.

An \emph{algebraic integer} is an algebraic number that is a root of a monic polynomial with integer coefficients. The algebraic integers belonging to $E$ form a ring, denoted by $\OO_E$, called the \emph{ring of integers} of $E$. A basic fact from algebraic number theory is that there are elements $\omega_1,\ldots,\omega_D\in\OO_E$ such that every element of $\OO_E$ can be written uniquely as
\[
 n_1\omega_1+\cdots+n_D\omega_D,\qquad n_1,\ldots,n_D\in\Z.
\]
We will also use certain subrings $R\subseteq\OO_E$ with the same $D$-dimensional integer-lattice structure. More precisely, an \emph{order} in $E$ is a subring $R\subseteq\OO_E$ containing $1$ for which there are elements $\omega_1,\ldots,\omega_D\in R$ such that every element of $R$ has a unique representation
\[
 n_1\omega_1+\cdots+n_D\omega_D,\qquad n_1,\ldots,n_D\in\Z,
\]
and $\omega_1,\ldots,\omega_D$ form a basis of $E$ when rational coefficients are allowed. We call $\omega_1,\ldots,\omega_D$ a \emph{$\Z$-basis} of $R$. The ring $\OO_E$ itself is an order. Another example is $\Z[\alpha]$ whenever $\alpha$ is an algebraic integer with $E=\Q(\alpha)$.

For $z\in E$, its \emph{norm} from $E$ to $\Q$ is
\[
 N_{E/\Q}(z)=\prod_{\sigma}\sigma(z),
\]
where the product runs over all $D$ embeddings $\sigma:E\hookrightarrow\C$. Equivalently, consider multiplication by $z$ as a $\Q$-linear map from $E$ to itself. Its determinant is $N_{E/\Q}(z)$. If $z\in R$, then multiplication by $z$ sends $R$ to itself, so in a $\Z$-basis it is represented by an integer matrix. Consequently,
\[
 N_{E/\Q}(z)\in\Z.
\]
Moreover, $N_{E/\Q}(z)=0$ only when $z=0$. These standard facts about number fields and their rings of integers are discussed in \cite[Chapters~2--3]{Milne}.

We will reduce elements of $R$ modulo suitable prime numbers. A \emph{reduction of $R$ at $p$} is a ring homomorphism
\[
 \varphi_p:R\longrightarrow\F_p
\]
with $\varphi_p(1)=1$. Such a homomorphism is automatically surjective, since its image contains the prime subfield $\F_p$. On ordinary integers this is necessarily the usual reduction modulo $p$. The key point for us is that if an element of $R$ is small under every embedding of $E$ and reduces to zero, then it must already have been zero.

\begin{example}\label{ex:sqrt2}
Take $E=\Q(\sqrt2)$ and $R=\Z[\sqrt2]$, with $\Z$-basis $1,\sqrt2$. The two embeddings $E\hookrightarrow\C$ send $\sqrt2$ to $\sqrt2$ and $-\sqrt2$, respectively. Thus
\[
 N_{E/\Q}(a+b\sqrt2)=(a+b\sqrt2)(a-b\sqrt2)=a^2-2b^2.
\]
If $t\in\F_p$ satisfies $t^2=2$, then
\[
 a+b\sqrt2\longmapsto a+bt
\]
defines a reduction $\varphi_p:R\to\F_p$. If $z=a+b\sqrt2$ lies in the kernel, then its norm is divisible by $p$. Therefore, if both $|a+b\sqrt2|$ and $|a-b\sqrt2|$ are smaller than $\sqrt p$, then $|N_{E/\Q}(z)|<p$. Since the norm is an integer divisible by $p$, it must be zero, and hence $z=0$.
\end{example}

The same argument works in every number field.

\begin{lemma}[Exactness from small conjugates]\label{lem:smallconjugates}
Let $\varphi_p:R\to\F_p$ be a reduction at $p$. If $z\in R$ satisfies $\varphi_p(z)=0$ and
\[
 |\sigma(z)|<p^{1/D}
\]
for every embedding $\sigma:E\hookrightarrow\C$, then $z=0$.
\end{lemma}
\begin{proof}
Fix a $\Z$-basis $\omega_1,\ldots,\omega_D$ of $R$, and let $M_z$ be the $D\times D$ integer matrix of multiplication by $z$ in this basis. Thus
\[
 z\omega_j=\sum_{i=1}^D(M_z)_{ij}\omega_i,
 \qquad\det M_z=N_{E/\Q}(z).
\]
Consider the row vector
\[
 v=\bigl(\varphi_p(\omega_1),\ldots,\varphi_p(\omega_D)\bigr)\in\F_p^D.
\]
It is nonzero: otherwise $\varphi_p$ would vanish on every integer linear combination of the basis elements, hence on all of $R$, contradicting $\varphi_p(1)=1$.

Apply $\varphi_p$ to the identity for $z\omega_j$. Since $\varphi_p(z)=0$, we obtain
\[
 0=\sum_{i=1}^D(M_z)_{ij}\varphi_p(\omega_i)
\]
for every $j$, or $vM_z=0$ over $\F_p$. Because $v\ne0$, the matrix $M_z$ is singular modulo $p$, and therefore $p\mid\det M_z=N_{E/\Q}(z)$. On the other hand,
\[
 |N_{E/\Q}(z)|=\prod_{\sigma:E\hookrightarrow\C}|\sigma(z)|<p.
\]
The norm is an integer divisible by $p$, so it must be zero. Hence $z=0$.
\end{proof}

\subsection{Exact lifting from a box}\label{subsec:box}
Fix a $\Z$-basis $\omega_1,\ldots,\omega_D$ of $R$ and a finite set $\Lambda\subseteq E\setminus\{0,1\}$. Write each $\lambda\in\Lambda$ as $\lambda=t_\lambda/d_\lambda$ with $t_\lambda\in R$ and a positive integer $d_\lambda$. This is possible because the basis of $R$ spans $E$ over $\Q$: write $\lambda$ in that basis and clear the rational denominators. When $p$ divides no $d_\lambda$, a reduction $\varphi_p$ extends to these elements by
\[
 \bar\lambda:=\frac{\varphi_p(t_\lambda)}{d_\lambda}\in\F_p.
\]
For an integer $H\ge1$ consider the box
\begin{equation}\label{eq:box}
 B_H=\left\{\sum_{a=1}^D n_a\omega_a:0\le n_a<H\right\},
 \qquad |B_H|=H^D.
\end{equation}
Since the basis and $\Lambda$ are fixed, there is a constant $C\ge1$ depending only on them such that for every embedding $\sigma$, all $u,v,w\in B_H$ and all $\lambda\in\Lambda$,
\begin{equation}\label{eq:boxbound}
 |\sigma(u-v)|\le CH,\qquad
 \left|\sigma\bigl(d_\lambda(w-u)+t_\lambda(u-v)\bigr)\right|\le CH.
\end{equation}
Indeed, $u-v$ has basis coefficients of absolute value less than $H$, so
$|\sigma(u-v)|\le H\sum_a|\sigma(\omega_a)|$; the second bound follows by multiplying the corresponding bounds for $w-u$ and $u-v$ by $d_\lambda$ and $|\sigma(t_\lambda)|$. There are only finitely many embeddings and coefficients to consider.

\begin{lemma}[Exact lifting]\label{lem:exactlifting}
Let $p$ be a prime dividing no $d_\lambda$, admitting a reduction $\varphi_p$ of $R$, and satisfying $p^{1/D}\ge4C$. Put
\begin{equation}\label{eq:H}
 H=\left\lfloor\frac{p^{1/D}}{2C}\right\rfloor.
\end{equation}
Then $\varphi_p$ is injective on $B_H$, and for all $u,v,w\in B_H$ and $\lambda\in\Lambda$,
\[
 \varphi_p(w)=(1-\bar\lambda)\varphi_p(u)+\bar\lambda\varphi_p(v)
 \quad\Longrightarrow\quad w=(1-\lambda)u+\lambda v.
\]
\end{lemma}
\begin{proof}
By \eqref{eq:boxbound} and \eqref{eq:H}, every embedding of $u-v$ and of $d_\lambda(w-u)+t_\lambda(u-v)$ has absolute value at most $CH\le p^{1/D}/2<p^{1/D}$. If $\varphi_p(u)=\varphi_p(v)$, then $\varphi_p(u-v)=0$ and Lemma~\ref{lem:smallconjugates} gives $u=v$.
If $\varphi_p(w)=(1-\bar\lambda)\varphi_p(u)+\bar\lambda\varphi_p(v)$, multiplying by $d_\lambda$ gives
\[
 \varphi_p\bigl(d_\lambda(w-u)+t_\lambda(u-v)\bigr)=0.
\]
Lemma~\ref{lem:smallconjugates} gives $d_\lambda(w-u)+t_\lambda(u-v)=0$ in $R$; dividing by $d_\lambda$ in $E$ yields $w=(1-\lambda)u+\lambda v$.
\end{proof}

\subsection{A large level set}\label{subsec:levelset}
For $v=(v_1,\ldots,v_m)\in B_H^m$, let
\[
 Q(v)=\sum_{j=1}^m v_j^2\in R.
\]
Write each $v_j$ in the fixed basis $\omega_1,\ldots,\omega_D$. Its coefficients have absolute value less than $H$, and the products $\omega_a\omega_b$ are fixed integer combinations of the basis elements. Thus each basis coefficient of $Q(v)$ is an integer of absolute value $O_R(mH^2)$, where the implied constant also depends on the chosen basis. Consequently there is a constant $C_2\ge1$, depending only on $R$ and its basis, such that each coefficient has at most $C_2mH^2$ possible values. It follows that $Q$ takes at most $(C_2mH^2)^D$ values on $B_H^m$. A largest level set $T\subseteq B_H^m$ of $Q$ therefore satisfies
\begin{equation}\label{eq:levelsize}
 |T|\ge\frac{H^{Dm}}{(C_2mH^2)^D}.
\end{equation}

\begin{lemma}\label{lem:noline}
Three pairwise distinct points of a level set of $Q$ never lie on an affine line over $E$: if $u,v,w\in T$ are pairwise distinct, then $w\ne(1-\lambda)u+\lambda v$ for every $\lambda\in E$.
\end{lemma}
\begin{proof}
Suppose $w=(1-\lambda)u+\lambda v$ with $Q(u)=Q(v)=Q(w)$. Expanding the squares,
\[
 Q(w)-(1-\lambda)Q(u)-\lambda Q(v)
 =-\lambda(1-\lambda)\sum_{j=1}^m(u_j-v_j)^2.
\]
The left side is zero, and $\lambda\notin\{0,1\}$ because $w\ne u,v$, so $\sum_j(u_j-v_j)^2=0$ in $E$. In the real embedding $E\subseteq\R$ this is a vanishing sum of squares of real numbers, so $u_j=v_j$ for all $j$, contradicting $u\ne v$.
\end{proof}

Only the one real embedding is used here, whereas Lemma~\ref{lem:smallconjugates} uses all embeddings; $E$ need not be totally real.

\begin{proposition}[Arithmetic sphere construction]\label{prop:arithmeticsphere}
Let $E\subseteq\R$, $R$ and $\Lambda$ be as above. There are constants $a\ge0$ and $b\ge1$ depending only on these objects such that the following holds for every prime $p$ that divides no $d_\lambda$, admits a reduction of $R$, and satisfies $p^{1/D}\ge4C$. For every $m\ge1$ there is a set $S_{p,m}\subseteq\F_{p^m}$ with
\begin{equation}\label{eq:spheresize}
 \log|S_{p,m}|\ge(m-2)\log p-am-D\log(bm)
\end{equation}
that contains no three pairwise distinct elements $x,y,z$ with $z=(1-\bar\lambda)x+\bar\lambda y$ for some $\lambda\in\Lambda$.
\end{proposition}
\begin{proof}
Let $H$ be as in \eqref{eq:H}, let $T$ be a largest level set of $Q$ on $B_H^m$, and let $S_{p,m}$ be the image of $T$ under coordinatewise reduction $B_H^m\to\F_p^m$ followed by the identification of $\F_p^m$ with $\F_{p^m}$ through a fixed $\F_p$-basis. Reduction is injective on $B_H$ by Lemma~\ref{lem:exactlifting}, so $|S_{p,m}|=|T|$.
Suppose $x,y,z\in S_{p,m}$ are pairwise distinct with $z=(1-\bar\lambda)x+\bar\lambda y$, and let $u,v,w\in T$ be their preimages, also pairwise distinct. Since $\bar\lambda\in\F_p$ and the identification is $\F_p$-linear, the relation holds in each of the $m$ coordinates. Lemma~\ref{lem:exactlifting} lifts it coordinatewise to $w=(1-\lambda)u+\lambda v$ in $E^m$, contradicting Lemma~\ref{lem:noline}.

For the size, \eqref{eq:H} and $p^{1/D}\ge4C$ give $p^{1/D}/(4C)\le H\le p^{1/D}$, so $H^{Dm}\ge p^m/(4C)^{Dm}$ and $H^{2D}\le p^2$. Substituting into \eqref{eq:levelsize},
\[
 \log|T|\ge m\log p-mD\log(4C)-2\log p-D\log(C_2m),
\]
which is \eqref{eq:spheresize} with $a=D\log(4C)$ and $b=C_2$.
\end{proof}

\subsection{From locator sets to codes}\label{subsec:localization}
The sets $S_{p,m}$ avoid a fixed finite list of affine ratios. To make them admissible we need two more things: every dependency on $r+2$ locators must contain a triple in one of these ratios, and suitable primes must exist at every scale. We package both requirements in a definition. The restriction on $m$ is the one in the Yekhanin--Dumer theorem.

\begin{definition}\label{def:localization}
Fix $q$ and $d$, with $K=q-1$ and $r=d-3$. A \emph{localization datum} for $(q,d)$ consists of a number field $E\subseteq\R$, an order $R\subseteq E$, a finite set $\Lambda\subseteq E\setminus\{0,1\}$ with fixed representations $\lambda=t_\lambda/d_\lambda$ as above, and a set of primes $\cP$, such that
\begin{enumerate}[label=(\roman*),leftmargin=2.2em]
\item $\cP$ contains a prime in $[Y,2Y]$ for every sufficiently large $Y$;
\item every sufficiently large $p\in\cP$ divides no $d_\lambda$ and admits a reduction $\varphi_p:R\to\F_p$;
\item for every such $p$, every prime $m>r!$, and every dependency on distinct locators $x_1,\ldots,x_{r+2}\in\F_{p^m}$, listed in any order, one has
$x_3=(1-\bar\lambda)x_1+\bar\lambda x_2$ for some $\lambda\in\Lambda$.
\end{enumerate}
\end{definition}

\begin{proposition}\label{prop:localization}
If $(q,d)$ admits a localization datum, then Theorem~\ref{thm:main} holds for $(q,d)$.
\end{proposition}
\begin{proof}
Let $p\in\cP$ be large, let $m>r!$ be prime, and let $S_{p,m}$ be the set of Proposition~\ref{prop:arithmeticsphere}. It is admissible: a dependency on its elements involves exactly $r+2$ locators by Lemma~\ref{lem:vandermonde}, and by (iii) three of them satisfy $x_3=(1-\bar\lambda)x_1+\bar\lambda x_2$ with $\lambda\in\Lambda$, which $S_{p,m}$ does not permit. Admissibility is inherited by subsets, so Lemma~\ref{lem:momentcode} gives, for every $n\le|S_{p,m}|$, a $q$-ary code of length $n$, minimum distance at least $d$, and redundancy at most $(rm+1)\log_q p$.

It remains to choose $p$ and $m$ in terms of $n$. Write $X=\log n$ and $M=\lceil\sqrt X\rceil+3$. By Bertrand's postulate, choose a prime $m$ with $M<m<2M$. For all sufficiently large $n$ we then have $m>r!$, and $m=\Theta(\sqrt X)$. Define $Y$ by
\[
 \log Y=\frac{X+am+D\log(bm)}{m-2}.
\]
As $n\to\infty$, also $Y\to\infty$. Condition (i) supplies $p\in\cP\cap[Y,2Y]$, and \eqref{eq:spheresize} gives $\log|S_{p,m}|\ge X$, that is, $|S_{p,m}|\ge n$. The redundancy of the resulting code, in natural-log units, is at most
\begin{align*}
 (rm+1)\log p
 &\le\frac{rm+1}{m-2}\bigl(X+am+D\log(bm)\bigr)+(rm+1)\log2\\
 &=rX+O_{q,d}(\sqrt X).
\end{align*}
Indeed, $(rm+1)/(m-2)=r+(2r+1)/(m-2)$ and $m=\Theta(\sqrt X)$, so the additional terms are $O(X/m+m+\log m)=O(\sqrt X)$. Exponentiating gives \eqref{eq:main}, as $r=d-3$.
\end{proof}

The field $E$, the set $\Lambda$, and hence $a,b,D$ are fixed before $p$ and $m$ grow. Requiring $m$ to be prime changes only the lower-order term. Condition (i), on the other hand, is essential to this calculation: an infinite set of permitted primes with arbitrarily large multiplicative gaps would not suffice.

\section{Distance five}\label{sec:five}
We now take $d=5$, so $r=2$: the checksums are the symbol sum and the first two moments, and by Lemma~\ref{lem:vandermonde} the only dependencies to worry about involve exactly four locators. Here the shapes are roots of quadratics and the auxiliary primes can be found in an arithmetic progression, so nothing beyond quadratic reciprocity is required. The section also serves as a template for the general case.

\subsection{The shapes}
A four-term dependency has coefficients $(a,b,c,e)\in\cD$, that is, nonzero integers of absolute value at most $K$ with $a+b+c+e=0$, and distinct locators $x_1,\ldots,x_4$ with
\[
 ax_1+bx_2+cx_3+ex_4=0,\qquad
 ax_1^2+bx_2^2+cx_3^2+ex_4^2=0.
\]
After normalizing $x_1=0$ and $x_2=1$, the remaining two locators $y,z$ satisfy $b+cy+ez=0$ and $b+cy^2+ez^2=0$. Multiplying the second equation by $e$ and substituting $ez=-(b+cy)$ from the first shows that $y$ is a root of
\begin{equation}\label{eq:quadratic}
 P_{a,b,c,e}(Y)=c(c+e)Y^2+2bcY+b(b+e).
\end{equation}
If $c+e\ne0$ this is a quadratic, and a short calculation using $b+c+e=-a$ gives its discriminant as
\begin{equation}\label{eq:discriminant}
 (2bc)^2-4c(c+e)b(b+e)=4abce.
\end{equation}
If $c+e=0$ the polynomial is $2bcY+b(b-c)$, of degree one with nonzero leading coefficient. In either case $y$ is a root of a nonzero polynomial of degree at most two whose coefficients are integers built from $a,b,c,e$.

There are finitely many tuples $(a,b,c,e)$ and hence finitely many polynomials \eqref{eq:quadratic}. Their real roots lie in the multiquadratic field
\[
 E=\Q(\sqrt{\ell_1},\ldots,\sqrt{\ell_h}),
\]
where $\ell_1,\ldots,\ell_h$ are the primes not exceeding $K$: a positive discriminant $4abce$ is four times a product of positive integers at most $K$, so its square root is an integer times a product of some of the $\sqrt{\ell_i}$, and the root in the linear case is rational. For $K=1$ we have $E=\Q$. Let $\Lambda\subseteq E\setminus\{0,1\}$ be the set of all real roots other than $0$ and $1$ of all the polynomials \eqref{eq:quadratic}, with representations $\lambda=t_\lambda/d_\lambda$ fixed as in Section~\ref{sec:sphere}.

\subsection{The field and the primes}
As an order in $E$, take 
\[
R=\Z[\sqrt{\ell_1},\ldots,\sqrt{\ell_h}].
\]
Each $\sqrt{\ell_i}$ is an algebraic integer, since it satisfies the monic polynomial $X^2-\ell_i$. As sums and products of algebraic integers are again
algebraic integers, we have $R\subseteq\mathcal O_E$.

Since no nonempty product of distinct primes is a square of a rational number, a standard fact about multiquadratic extensions gives $D=[E:\Q]=2^h$; see
\cite[Theorem~5.30 and Example~5.31(b)]{MilneFT}.
A $\Q$-basis of $E$ and a $\Z$-basis of $R$ are given by the $2^h$ products
\[
\omega_I=\prod_{i\in I}\sqrt{\ell_i},
\qquad I\subseteq\{1,\ldots,h\}.
\]
The multiplication rule is
\[
 \omega_I\omega_J=\left(\prod_{i\in I\cap J}\ell_i\right)\omega_{I\mathbin{\triangle}J}.
\]
All embeddings of $E$ are real. Put
\begin{equation}\label{eq:TK}
 T_K=8\prod_{\substack{\ell\le K\\\ell\text{ odd prime}}}\ell,
\end{equation}
and let $\cP$ be the set of primes $p\equiv-1\pmod{T_K}$.

\begin{lemma}\label{lem:quadraticprimes}
Let $p\in\cP$. Then $-1$ is a quadratic nonresidue modulo $p$, every positive integer at most $K$ is a quadratic residue, and $R$ admits a reduction at $p$. Moreover, $\cP$ contains a prime in $[Y,2Y]$ for all sufficiently large $Y$.
\end{lemma}
\begin{proof}
By the supplementary laws of quadratic reciprocity, $p\equiv 7\pmod 8$ implies that $-1$ is a quadratic nonresidue and $2$ is a quadratic residue modulo $p$. For an odd prime $\ell\le K$ we have $p\equiv-1\pmod\ell$ and $p\equiv3\pmod4$, so quadratic reciprocity gives
\[
 \left(\frac\ell p\right)
 =\left(\frac p\ell\right)(-1)^{(\ell-1)/2}
 =\left(\frac{-1}\ell\right)(-1)^{(\ell-1)/2}=1.
\]
Every positive integer at most $K$ is a product of such primes and of $2$, hence a residue. To define a reduction, choose a square root $s_i\in\F_p$ of each $\ell_i$ and send $\omega_I\mapsto\prod_{i\in I}s_i$; the multiplication rule shows that this is a ring homomorphism, and it is onto because it restricts to the natural map on $\Z$.
Finally, the prime number theorem for arithmetic progressions \cite[Theorem~4.12]{KedlayaAP} gives
\[
 \#\{p\le x:p\equiv-1\pmod{T_K}\}\sim\frac{x}{\phi(T_K)\log x},
\]
where $\phi$ is Euler's totient function. Comparing this at $Y$ and $2Y$ produces a prime of $\cP$ in $[Y,2Y]$ for large $Y$.
\end{proof}

\begin{lemma}\label{lem:quadraticlifts}
Let $p\in\cP$ with $p>4K$ and $p$ dividing no $d_\lambda$, and let $m$ be odd. Every dependency on distinct $x_1,\ldots,x_4\in\F_{p^m}$ satisfies $x_3=(1-\bar\lambda)x_1+\bar\lambda x_2$ for some $\lambda\in\Lambda$.
\end{lemma}
\begin{proof}
Normalize as above. The coordinate $y=(x_3-x_1)/(x_2-x_1)$ is a root of the reduction of $P_{a,b,c,e}$ modulo $p$, which is a nonzero polynomial of degree at most two over $\F_p$: its leading coefficient is a nonzero integer all of whose prime factors are at most $2K<p$. Hence $y$ has degree at most two over $\F_p$. But $y$ also lies in $\F_{p^m}$ with $m$ odd, so its degree divides $m$; therefore $y\in\F_p$. For prime $m>2$ this is also an immediate consequence of Lemma~\ref{lem:affineline}.

Suppose $c+e\ne0$. Since the quadratic has a root in $\F_p$, its discriminant $4abce$ is a square in $\F_p$. It is nonzero there, because $p$ is odd and exceeds $|a|,|b|,|c|,|e|$. If $abce<0$, then $4abce=-4|abce|$ is a nonresidue by Lemma~\ref{lem:quadraticprimes}, since $|abce|$ is a product of positive integers at most $K$ and $-1$ is a nonresidue; this is impossible. So $abce>0$, the two square roots of the discriminant in $\F_p$ are the reductions of $\pm2\sqrt{|abce|}\in E$, and the quadratic formula, valid since $p$ is odd and does not divide $c(c+e)$, shows that $y$ is the reduction of one of the two real roots of $P_{a,b,c,e}$. If $c+e=0$, then $y$ is the reduction of the rational root. In both cases $y\notin\{0,1\}$ because $x_3\ne x_1,x_2$, so $y=\bar\lambda$ for some $\lambda\in\Lambda$. Undoing the normalization gives
\[
 x_3=x_1+y(x_2-x_1)=(1-\bar\lambda)x_1+\bar\lambda x_2.\qedhere
\]
\end{proof}

\begin{proof}[Proof of Corollary~\ref{cor:fourfive} at distance five]
Every prime $m>r!=2$ is odd, so Lemmas~\ref{lem:quadraticprimes} and~\ref{lem:quadraticlifts} show that $(E,R,\Lambda,\cP)$ is a localization datum for $(q,5)$. Proposition~\ref{prop:localization} gives $\rho_q(n,5)\le2\log_q n+O_q(\sqrt{\log n})$. The lower bound is \eqref{eq:hamming} with $t=2$.
\end{proof}

\begin{remark}[Alternative Sidon-set route]\label{rem:sidon} There is a shorter route to the distance-five upper bound using $k$-fold Sidon sets. In the terminology of Cilleruelo and Timmons~\cite{CT}, a set $A\subseteq\mathbb Z$ is $k$-fold Sidon if every solution in $A$ to an equation
\[
    c_1x_1+c_2x_2+c_3x_3+c_4x_4=0,
    \qquad
    |c_i|\le k,\quad \sum_{i=1}^4 c_i=0,
\]
is trivial: the indices can be partitioned into groups such that the variables in each group are equal and the corresponding coefficients sum to zero. In particular, for every $2\le w\le 4$, there is no relation
\[
    \sum_{i=1}^w c_i x_i=0,
    \qquad
    0<|c_i|\le k,\quad \sum_{i=1}^w c_i=0,
\]
on pairwise distinct elements $x_i\in A$. For $w<4$, this follows by padding the equation with zero coefficients.

Cilleruelo and Timmons~\cite{CT} report that, for every fixed $k$, there exist sets $A\subseteq[1,N]$ with this property and
\[
|A|\ge N^{1/2}\exp(-O_k(\sqrt{\log N})).
\]
Importantly, a single set $A$ avoids all such equations simultaneously, rather than just one fixed equation. They attribute this extension of Ruzsa's four-variable construction~\cite[Theorem~7.3]{Ruzsa} to Axenovich and Verstra\"ete (personal communication and unpublished work).

Taking $k=K=q-1$, choose $n$ distinct locators $a_i$ from such a set and use the two checks
\[
 \sum_i u_i\pmod{4K+1},\qquad
 \sum_i a_i u_i\pmod{4KN+1}.
\]
Suppose two distinct words in one fiber differ in $w\le 4$ positions, with distinct locators $x_1,\ldots,x_w\in A$ and nonzero coordinate differences $\delta_1,\ldots,\delta_w$, where $|\delta_i|\le K$. Since the checksum moduli exceed $4K$ and $4KN$, respectively, the congruences lift to exact integer equalities
\[
    \sum_{i=1}^w\delta_i=0,
    \qquad
    \sum_{i=1}^w\delta_i x_i=0.
\]
The first equality rules out $w=1$, and for $2\le w\le 4$ the second contradicts the $K$-fold Sidon property stated above. Inverting the displayed size bound gives $N\le n^2\exp(O_K(\sqrt{\log n}))$, and hence redundancy $2\log_q n+O_q(\sqrt{\log n})$.

This would give a short proof of the distance-five bound if the stated Sidon-set estimate were available. However, the required simultaneous construction was
attributed to unpublished work. More recently, Fabian, Morris, and Szab\'o~\cite[Section~7.1]{FMS} reported that they could not recover its proof and pointed
out a potential difficulty in simultaneously avoiding equations for which the product of the four coefficients is not a square. Our proof does not rely on this unpublished result: the localization argument above establishes the distance-five bound directly and extends to arbitrary $d$.
\end{remark}

\begin{remark}\label{rem:quadraticfrobenius}
The congruence condition on $p$ has a Galois-theoretic meaning that anticipates Section~\ref{sec:general}. All roots of the polynomials \eqref{eq:quadratic}, real or not, lie in
$L=\Q(\sqrt{-1},\sqrt{\ell_1},\ldots,\sqrt{\ell_h})$.
For a prime $p\equiv-1\pmod{T_K}$, the $p$th-power map on their reductions fixes each $\sqrt{\ell_i}$ and negates $\sqrt{-1}$. This is the action of complex conjugation on the original roots. Consequently, roots whose reductions lie in $\F_p$ have real lifts. Section~\ref{sec:general} obtains this same property for general polynomials, using Chebotarev's theorem in place of the prime number theorem for arithmetic progressions.
\end{remark}

\section{General distance}\label{sec:general}
We return to arbitrary $d\ge4$, with $K=q-1$ and $r=d-3$. The Yekhanin--Dumer theorem already puts the normalized locators in $\F_p$ when $m>r!$ is prime. We only need to show that the first normalized coordinate is the reduction of one of finitely many fixed real algebraic numbers. We do this with a one-variable polynomial identity, followed by a choice of primes.

\subsection{One polynomial for the first coordinate}\label{subsec:onepolynomial}
Recall the polynomials $f_1,\ldots,f_r$ from \eqref{eq:shapes} and $\Delta$ from \eqref{eq:Delta}. A tuple $y$ is a shape exactly when it extends to a solution of
\[
 f_1(y)=\cdots=f_r(y)=0,\qquad z\Delta(y)-1=0.
\]

\begin{lemma}[A polynomial relation that survives reduction]\label{lem:polynomialrelation}
For every $\delta\in\cD$ there are a monic polynomial $P_\delta(T)\in\Q[T]$ and a positive integer $A_\delta$ such that, for every prime $p\nmid A_\delta$, the polynomial $P_\delta$ can be reduced modulo $p$ and every shape $y$ over $\Fpbar$ satisfies
\[
 \overline{P_\delta}(y_1)=0.
\]
The constant polynomial $P_\delta=1$ is allowed; in that case there are no shapes in these characteristics.
\end{lemma}
\begin{proof}
The proof of Lemma~\ref{lem:finiteshapes}, applied to $y_1$, gives a nonzero polynomial $P_\delta(T)\in\Q[T]$ for which $P_\delta(y_1)$ is a polynomial combination, over $\Q$, of $f_1,\ldots,f_r,z\Delta-1$. Dividing by its leading coefficient makes $P_\delta$ monic. Clearing all denominators in this identity gives a positive integer $A_\delta$ and polynomials $U_1,\ldots,U_r,V\in\Z[y_1,\ldots,y_r,z]$ such that $A_\delta P_\delta\in\Z[T]$ and
\begin{equation}\label{eq:certificate}
 A_\delta P_\delta(y_1)
 =\sum_{j=1}^r U_j(y,z)f_j(y)+V(y,z)\bigl(z\Delta(y)-1\bigr).
\end{equation}
For $p\nmid A_\delta$, reduce this identity modulo $p$ and substitute a shape, with $z=1/\Delta(y)$. The right side vanishes. Since $A_\delta$ is invertible modulo $p$, the desired relation follows.
\end{proof}

The point is the identity \eqref{eq:certificate}, not merely the fact that $P_\delta(y_1)$ vanishes at all complex shapes. The identity can be reduced modulo $p$ directly. We do not need to show that every shape modulo $p$ is the reduction of a characteristic-zero shape. Nor do we need every root of $P_\delta$ to occur as a coordinate of a shape: additional roots only add finitely many ratios to the list we will avoid.

\subsection{Choosing primes with real root lifts}\label{subsec:realroots}
Let $\Omega\subseteq\C$ be the union of the sets of complex roots of the polynomials $P_\delta$, where each distinct root is included only once, regardless of its multiplicity or how many polynomials contain it. Let $L\subseteq\C$ be their splitting field over $\Q$, that is, the field generated by all these roots. If $\Omega$ is empty, take $L=\Q$. Since the polynomials have rational coefficients, $L/\Q$ is a finite Galois extension. Its Galois group
\[
 G=\Gal(L/\Q)
\]
is the finite group of field automorphisms of $L$ that fix $\Q$; its elements permute the roots in $\Omega$. Complex conjugation restricts to one such automorphism $c\in G$. Its fixed field is
\begin{equation}\label{eq:realfield}
 E=\{\alpha\in L:c(\alpha)=\alpha\}=L\cap\R.
\end{equation}
We use this particular embedding $E\subseteq\R$ and take $R=\OO_E$. Other embeddings of $E$ need not be real.

\smallskip\noindent\textbf{Reducing algebraic numbers.}
Choose a positive integer $M$ such that $M\alpha\in\OO_L$ for every $\alpha\in\Omega$. Such an $M$ exists: a $\Z$-basis of $\OO_L$ spans $L$ over $\Q$, so one can clear the finitely many rational denominators of these roots in that basis.

To reduce elements of $\OO_L$ modulo a prime $p$, we use a surjective ring homomorphism
\[
\psi:\OO_L\longrightarrow\kappa,
\] 
where $\kappa$ is a finite field of characteristic $p$. Such maps exist for every $p$. Indeed, since $\OO_L$ has a $\Z$-basis of size $[L:\Q]$, the quotient
$\OO_L/p\OO_L$ is a finite ring. Taking its quotient by a maximal ideal gives a finite field $\kappa$, necessarily of characteristic $p$. The composition
of the quotient maps gives $\psi$. Different choices of maximal ideal may give different reductions. In algebraic number theory, the kernel of $\psi$ is called a prime ideal of $\OO_L$ lying
above $p$. Unlike the reductions in Section~\ref{sec:sphere}, the target $\kappa$ may be a proper extension of $\F_p$. If $p\nmid M$, we also define
\[
 \bar\alpha=\frac{\psi(M\alpha)}{M}\in\kappa\qquad(\alpha\in\Omega).
\]
This is just the denominator-clearing convention used for $\lambda$ in Section~\ref{subsec:box}.

Outside a finite set of primes, distinct roots in $\Omega$ have distinct reductions, for every choice of $\psi$. Here is a direct verification. The product
\[
 J=\prod_{\substack{\alpha,\beta\in\Omega\\\alpha\ne\beta}}
       M(\alpha-\beta)
\]
over ordered pairs is a nonzero integer. Indeed, every factor is an algebraic integer, and every automorphism in $G$ merely permutes the factors, so $J$ is rational as well as an algebraic integer. If two distinct roots have equal reductions, one factor reduces to zero and hence $p\mid J$. Excluding the prime divisors of $MJ$ prevents all such collisions. Empty products are understood to be $1$.

\smallskip\noindent\textbf{Frobenius describes the $p$th-power map.}
We use the following standard fact about number fields \cite[Chapter~8]{Milne}. Outside a finite set of primes (the ramified primes), each reduction $\psi$ has an associated automorphism $g_\psi\in G$, called its \emph{Frobenius automorphism}, characterized by
\begin{equation}\label{eq:frobenius}
 \psi\bigl(g_\psi(\alpha)\bigr)=\psi(\alpha)^p
 \qquad(\alpha\in\OO_L).
\end{equation}
Thus $g_\psi$ describes, in characteristic zero, the action of the familiar map $a\mapsto a^p$ after reduction. As $\psi$ varies over the reductions at the same prime $p$, the automorphisms $g_\psi$ run through one conjugacy class in $G$. In particular, if this class contains $c$, we can choose a reduction for which $g_\psi=c$ exactly.

\smallskip\noindent\textbf{Choosing the primes.}
Chebotarev's density theorem \cite[Theorem~22.5]{KedlayaCheb} says that the set of primes whose Frobenius class is the class of $c$ has counting function
\begin{equation}\label{eq:chebotarev}
 \#\{p\le x:g_\psi=c\text{ for some reduction }\psi\}
 \sim\frac{|[c]|}{|G|}\frac{x}{\log x},
\end{equation}
where $[c]=\{hch^{-1}:h\in G\}$. Exclude the finitely many ramified primes, the prime divisors of $MJ$ and of all the $A_\delta$, and the primes at most $(r+2)K$. Let $\cP$ be the remaining primes counted in \eqref{eq:chebotarev}. For each $p\in\cP$, fix a reduction
\[
 \psi_p:\OO_L\longrightarrow\kappa_p
 \quad\text{with}\quad g_{\psi_p}=c.
\]
Removing finitely many primes does not affect the asymptotic formula. Comparing it at $Y$ and $2Y$ shows that $\cP$ contains a prime in $[Y,2Y]$ for every sufficiently large $Y$.

Define
\[
 \Lambda=(\Omega\cap\R)\setminus\{0,1\}\subseteq E.
\]
For $\lambda\in\Lambda$, take $t_\lambda=M\lambda\in\OO_E$ and $d_\lambda=M$. Let $\varphi_p$ be the restriction of $\psi_p$ to $\OO_E$. The next lemma supplies both the reduction and the real ratios needed in Section~\ref{sec:sphere}.

\begin{lemma}[Real lifts of roots]\label{lem:realrootlifts}
For every $p\in\cP$:
\begin{enumerate}[label=(\alph*),leftmargin=2.2em]
\item $\varphi_p$ takes values in $\F_p$ and is a reduction of $\OO_E$ at $p$;
\item every $t\in\F_p$ satisfying $\overline{P_\delta}(t)=0$, for any $\delta\in\cD$, is the reduction of a real root of $P_\delta$. If $t\notin\{0,1\}$, then $t=\bar\lambda$ for some $\lambda\in\Lambda$.
\end{enumerate}
\end{lemma}
\begin{proof}
For (a), let $\alpha\in\OO_E$. Since $c(\alpha)=\alpha$, \eqref{eq:frobenius} gives
\[
 \psi_p(\alpha)^p=\psi_p\bigl(c(\alpha)\bigr)=\psi_p(\alpha).
\]
The elements of the finite field $\kappa_p$ fixed by the $p$th-power map are exactly $\F_p$. Hence $\varphi_p$ takes values there. It is a unital ring homomorphism and is surjective because it reduces ordinary integers modulo $p$.

For (b), factor the monic polynomial $P_\delta$ over $L$ and reduce this factorization. All its roots, with their multiplicities, have defined reductions in $\kappa_p$. It follows that a root $t\in\F_p\subseteq\kappa_p$ of $\overline{P_\delta}$ is the reduction $\bar\lambda$ of some root $\lambda\in\Omega$ of $P_\delta$. The Frobenius identity extends to these roots by clearing the denominator $M$, since $p\nmid M$ and $c(M)=M$. Thus
\[
 \overline{c(\lambda)}=\bar\lambda^{\,p}=t^p=t=\bar\lambda.
\]
Both $\lambda$ and $c(\lambda)$ belong to $\Omega$, and distinct elements of $\Omega$ have distinct reductions. Therefore $c(\lambda)=\lambda$, so $\lambda$ is real. If $t\ne0,1$, then $\lambda\ne0,1$, and $\lambda\in\Lambda$.
\end{proof}

\subsection{The localization datum}
\begin{proposition}\label{prop:generaldatum}
The data $(E,\OO_E,\Lambda,\cP)$ above form a localization datum for $(q,d)$.
\end{proposition}
\begin{proof}
Condition (i) of Definition~\ref{def:localization} follows from \eqref{eq:chebotarev}. For (ii), the primes in $\cP$ do not divide $M=d_\lambda$, and Lemma~\ref{lem:realrootlifts}(a) provides the reductions.

For (iii), let $p\in\cP$, let $m>r!$ be prime, and let $x_1,\ldots,x_{r+2}\in\F_{p^m}$ carry a dependency with coefficients $\delta\in\cD$. The normalized locators
\[
 y_i=\frac{x_{i+2}-x_1}{x_2-x_1},\qquad 1\le i\le r,
\]
form a shape for $\delta$. Lemma~\ref{lem:affineline} gives $y_1\in\F_p$. Lemma~\ref{lem:polynomialrelation} gives $\overline{P_\delta}(y_1)=0$. Since $y_1\ne0,1$, Lemma~\ref{lem:realrootlifts}(b) gives $y_1=\bar\lambda$ for some $\lambda\in\Lambda$. Undoing the normalization,
\[
 x_3=x_1+y_1(x_2-x_1)=(1-\bar\lambda)x_1+\bar\lambda x_2.\qedhere
\]
\end{proof}

\begin{proof}[Proof of Theorem~\ref{thm:main}]
Combine Propositions~\ref{prop:generaldatum} and~\ref{prop:localization}.
\end{proof}

The roles of the ingredients are now separate. Yekhanin--Dumer puts the dependency on a line over $\F_p$. A polynomial identity restricts the needed line ratio to the roots of a fixed polynomial. Chebotarev makes that ratio the reduction of a real algebraic number. The arithmetic sphere excludes it without an additional checksum. At no point do we need to lift the other coordinates of the shape.

\subsection{Two calculations at distance six}
The following examples illustrate the effect of the choice of primes on the normalized equations.

First take the coefficient pattern $(-4,1,1,1,1)$, available when $K\ge4$. Send the locators with coefficients $-4$ and $1$ to $0$ and $1$. The remaining three locators then have first three power sums $-1,-1,-1$, so by Newton's identities they are the roots of
\[
 X^3+X^2+X+1=(X+1)(X^2+1).
\]
The shapes are the permutations of $(-1,i,-i)$. For $p>4$ they exist over $\F_p$ when $p\equiv1\pmod4$, whereas for $p\equiv3\pmod4$ they are absent from every odd-degree extension of $\F_p$. At the level of individual roots, $-1$ is real and $i,-i$ are exchanged by complex conjugation. Choosing $p\equiv3\pmod4$ makes the $p$th-power map act on their reductions in exactly the same way. In particular, neither of the nonreal roots has its reduction in $\F_p$.

This example also distinguishes a real ratio from an entirely real shape: a shape may have a real coordinate, such as $-1$, even though its other coordinates are nonreal. Our proof only uses the one coordinate needed for a triple of locators.

For a real example take $q=3$ and $d=6$. Up to reordering and a global change of sign, the five-term patterns with coefficients in $\{\pm1,\pm2\}$ and zero sum are
\[
 (1,1,1,-1,-2)\qquad\text{and}\qquad(1,1,2,-2,-2).
\]
In the first, normalize the two negatively weighted locators to $0$ and $1$. The other three then have their first three power sums equal to $2$, so their monic root polynomial is $X(X-1)^2$ and they cannot be distinct; this pattern has no shapes. In the second, the normalized equations read
\[
 x+y+2z=2,\qquad x^2+y^2+2z^2=2,\qquad x^3+y^3+2z^3=2.
\]
The first two give $x+y=2-2z$ and $xy=1-4z+3z^2$, and substituting into the third yields $6z(z-1)(2z-1)=0$. Distinctness forces $z=1/2$ and hence
\begin{equation}\label{eq:ternarysix}
 (x,y,z)=\left(\frac{1+\sqrt2}{2},\frac{1-\sqrt2}{2},\frac12\right)
\end{equation}
up to interchanging $x$ and $y$. These computations remain valid in characteristic $p>10$.

Consequently, when $2$ is a nonsquare in $\F_p$ and $m$ is odd, $\F_{p^m}$ contains no five-point dependency of either ternary pattern on distinct locators. The whole field can then serve as the locator set, giving coefficient three with no sphere at all. This illustrates the mechanism; it does not improve the known ternary distance-six bounds \cite[Section~II]{YD}.

\section{Limitations and questions}\label{sec:limitations}
\subsection{Why the argument stops at $d-3$ moments}
Suppose only $k$ nonconstant moments are kept and a dependency on $s$ distinct locators is considered. After normalizing two locators there remain $s-2$ unknowns. For $k\le s-2$ the Jacobian of the $k$ moment equations has rank $k$ at every point with distinct coordinates, since any $k$ of its columns form a nonzero scaled Vandermonde determinant, so each nonempty component of the solution locus has dimension $s-2-k$. Our argument uses the boundary case $k=s-2=d-3$, where the locus is finite. With fewer checks the configurations, when they exist at all, form positive-dimensional families. Finiteness can then no longer supply the one-variable polynomial relations used in Section~\ref{sec:general}, and the affine-line theorem of Section~\ref{sec:moments} no longer follows from a bounded number of shapes. This is a limitation of the method rather than a lower bound for unrestricted codes.

\subsection{Questions}
Can the Hamming coefficient be reached for every fixed distance over every fixed alphabet? And how far can explicit nonlinear constructions separate from linear ones beyond distance four? Any attempt to push the real-lifting strategy below $d-3$ moments will have to cope with positive-dimensional configurations or find a different way to constrain short dependencies.


\begin{thebibliography}{99}
\bibitem{BK}
M.~Bateman and N.~H.~Katz, \emph{New bounds on cap sets}, Journal of the American Mathematical Society \textbf{25} (2012), no.~2, 585--613. arXiv:1101.5851.

\bibitem{Behrend}
F.~A.~Behrend, \emph{On sets of integers which contain no three terms in arithmetical progression}, Proceedings of the National Academy of Sciences of the United States of America \textbf{32} (1946), no.~12, 331--332. DOI: 10.1073/pnas.32.12.331.

\bibitem{CT}
J.~Cilleruelo and C.~Timmons, \emph{$k$-fold Sidon sets}, Electronic Journal of Combinatorics \textbf{21} (2014), no.~4, Paper P4.12. DOI: 10.37236/3860. arXiv:1310.5374.

\bibitem{CLO}
D.~A.~Cox, J.~Little, and D.~O'Shea, \emph{Ideals, Varieties, and Algorithms: An Introduction to Computational Algebraic Geometry and Commutative Algebra}, 3rd ed., Springer, New York, 2007.

\bibitem{CLP}
E.~Croot, V.~F.~Lev, and P.~P.~Pach, \emph{Progression-free sets in $\Z_4^n$ are exponentially small}, Annals of Mathematics \textbf{185} (2017), no.~1, 331--337. DOI: 10.4007/annals.2017.185.1.7.

\bibitem{Dumer}
I.~Dumer, \emph{Nonbinary double-error-correcting codes designed by means of algebraic varieties}, IEEE Transactions on Information Theory \textbf{41} (1995), no.~6, 1657--1666. DOI: 10.1109/18.476238.

\bibitem{EG}
J.~S.~Ellenberg and D.~Gijswijt, \emph{On large subsets of $\F_q^n$ with no three-term arithmetic progression}, Annals of Mathematics \textbf{185} (2017), no.~1, 339--343. DOI: 10.4007/annals.2017.185.1.8. arXiv:1605.09223.

\bibitem{EHPS}
C.~Elsholtz, Z.~Hunter, L.~Proske, and L.~Sauermann, \emph{Improving Behrend's construction: Sets without arithmetic progressions in integers and over finite fields}, preprint, 2024. arXiv:2406.12290.

\bibitem{FMS} D. Fabian, P. Morris, and T. Szab\'o, \emph{Slow graph bootstrap percolation III: Chain constructions}, preprint, 2025. arXiv:2508.03835.

\bibitem{KedlayaAP}
K.~S.~Kedlaya, \emph{Notes on analytic number theory}, Chapter~4, ``Primes in arithmetic progressions,'' especially Theorem~4.12, online lecture notes. \url{https://kskedlaya.org/ant/chap-primes-in-ap.html} (accessed October~7, 2026).

\bibitem{KedlayaCheb}
K.~S.~Kedlaya, \emph{Notes on analytic number theory}, Chapter~22, ``Artin $L$-functions and the Chebotaryov density theorem,'' especially Theorem~22.5, online lecture notes. \url{https://kskedlaya.org/ant/chap-artin.html} (accessed October~7, 2026).

\bibitem{LX}
S.~Liu and C.~Xing, \emph{Nonlinear codes with low redundancy}, in Proceedings of the 2024 IEEE International Symposium on Information Theory (ISIT), Athens, Greece, 2024, pp.~705--710. DOI: 10.1109/ISIT57864.2024.10619117. Extended version: arXiv:2310.14219.

\bibitem{Meshulam}
R.~Meshulam, \emph{On subsets of finite abelian groups with no 3-term arithmetic progressions}, Journal of Combinatorial Theory, Series~A \textbf{71} (1995), no.~1, 168--172. DOI: 10.1016/0097-3165(95)90024-1.

\bibitem{Milne}
J.~S.~Milne, \emph{Algebraic Number Theory}, version~3.08, July~19, 2020, 166~pp. \url{https://www.jmilne.org/math/CourseNotes/ANT.pdf}.

\bibitem{MilneFT} J.~S. Milne, \emph{Fields and Galois Theory}, version 5.10, 2022. \url{https://www.jmilne.org/math/CourseNotes/FT.pdf}.

\bibitem{Ruzsa}
I.~Z.~Ruzsa, \emph{Solving a linear equation in a set of integers I}, Acta Arithmetica \textbf{65} (1993), no.~3, 259--282. DOI: 10.4064/aa-65-3-259-282.

\bibitem{YD}
S.~Yekhanin and I.~Dumer, \emph{Long nonbinary codes exceeding the Gilbert--Varshamov bound for any fixed distance}, IEEE Transactions on Information Theory \textbf{50} (2004), no.~10, 2357--2362. DOI: 10.1109/TIT.2004.834744. arXiv:cs/0406039.
\end{thebibliography}
\end{document}